\documentclass[10pt,a4paper]{extarticle}
\usepackage[margin=2.5cm]{geometry}

\usepackage[T1]{fontenc}
\usepackage{microtype}

\usepackage{graphicx,xcolor}
\definecolor{newgreen}{RGB}{0,158,115}

\usepackage{amsmath,amsfonts,amssymb,amsthm,mathtools,bbm}
\usepackage{esint}
\usepackage{braket}
\usepackage{dsfont}
\numberwithin{equation}{section}

\usepackage[shortlabels]{enumitem}
\usepackage{booktabs}
\usepackage[font=small,margin=0.1\linewidth]{caption}
\usepackage{subcaption}
\usepackage{float}

\usepackage{comment}

\usepackage[colorlinks=true,hyperindex=true]{hyperref}
\hypersetup{
  citecolor=newgreen,
  linkcolor=newgreen,
  urlcolor=newgreen,
  pdfnewwindow=true,
  pdfstartview={FitH}
}

\usepackage{nameref}

\newtheorem{theorem}{Theorem}[section]
\newtheorem{proposition}[theorem]{Proposition}
\newtheorem{lemma}[theorem]{Lemma}
\newtheorem{corollary}[theorem]{Corollary}

\newtheorem{definition}[theorem]{Definition}
\newtheorem*{claim*}{Claim}

\theoremstyle{definition}
\newtheorem{remark}[theorem]{Remark}

\def\rr{{\mathds R}}

\def\nn{{\mathds N}}

\def\cA{{\cal A}}

\def\cF{{\cal F}}

\def\E{{\mathds{E}}}
\def\L{{\mathds{L}}}
\def\P{{\mathds{P}}}

\makeatletter
\providecommand*{\diff}{\@ifnextchar^{\DIfF}{\DIfF^{}}}
\def\DIfF^#1{%
  \mathop{\mathrm{\mathstrut d}}\nolimits^{#1}\gobblespace}
\def\gobblespace{\futurelet\diffarg\opspace}
\def\opspace{%
  \let\DiffSpace\!%
  \ifx\diffarg(\let\DiffSpace\relax\else
  \ifx\diffarg\let\DiffSpace\relax\else
  \ifx\diffarg\{\let\DiffSpace\relax\fi\fi\fi\DiffSpace}
\makeatother

\newcommand{\Exp}[1]{\mathrm{e}^{#1}}
\DeclareMathOperator{\supp}{supp}

\begin{document}

\title{Stability of the Shannon--McMillan--Breiman Theorem
under Sublinear Parsings}

\author{ Raphaël\ Grondin }
\date{}

\maketitle

\begin{center}
  \begin{minipage}[b]{0.5\textwidth}
    \small
    \centering
     Columbia University \\
    Department of Mathematics  \\
   New York NY, United States
  \end{minipage}%
\end{center}

\begin{abstract}
We establish a stability result for the Shannon--McMillan--Breiman theorem on
the one--sided finite shift space. For any shift--invariant probability measure
\(\P\) and any data--dependent parsing whose number of blocks
\(c_N\) satisfies \(c_N=o(N)\) almost surely, we show that the normalized sum of
the negative log--likelihoods of the parsing blocks converges \(\P\)--almost
surely and in \(\L^1(\P)\) to the entropy rate function \(h_\P\).

Equivalently, we obtain a new structural result for shift-invariant probability
measures on the one-sided finite shift space, in the form of an approximate
factorization of cylinder probabilities. Namely, for any sublinear parsing
\(x_1^N=w_1^{(N)}\dots w_{c_N}^{(N)}\),
\[
\P([x_1^N])
=
\Bigg(\prod_{i=1}^{c_N} \P([w_i^{(N)}])\Bigg)
\exp\{o(N)\},
\qquad \P\text{-a.s. and in }\L^1(\P).
\]
We further show that the stability result persists under subextensive perturbations
of the parsing blocks, and that the sublinearity of \(c_N\) is the sharp threshold
for validity at this level of generality, via a direct counterexample.
\end{abstract}

\tableofcontents

\section{Introduction}

Let $\cA$ be a finite alphabet and let $\Omega=\cA^{\nn}$ be the one--sided shift
space, equipped with the $\sigma$--field $\cF$ generated by the cylinder sets
$[x_1^N]:=\{y\in\Omega:y_1^N=x_1^N\}$, where $x_1^N=(x_1,x_2,\dots ,x_N)$ for
$x=(x_{k})_{k\geq 1}\in \Omega$.
Throughout, $\P$ denotes a shift--invariant probability measure on
$(\Omega,\cF)$, and $T:\Omega\to\Omega$, $(Tx)_k=x_{k+1}$, the left shift.

\medskip

The Shannon--McMillan--Breiman (SMB) theorem is a cornerstone of information theory
and ergodic theory.
It asserts that there exists a shift--invariant function $h_\P\in \L^1(\P)$ (the
entropy rate function) such that
\[
    \lim_{N\to\infty}-\frac1N \log \P([x_1^N])=h_\P(x),
    \qquad \P\text{-a.s.\ and in }\L^1(\P),
\]
and whose expectation $\E_\P(h_\P)=h(\P)$ is the entropy rate of $\P$
(see \eqref{eq:h} in Section~\ref{sec:notation}).
If $\P$ is ergodic, then $h_\P=h(\P)$ holds $\P$--a.s. (see, e.g.,
\cite{McM53,Breiman57,Brei60,SH96}).

\medskip

In many applications, however, one does not work directly with full blocks of the form
$x_1^N$, but rather with variable--length \emph{parsings} of these blocks.
In full generality, a parsing of $x_1^N$ is a decomposition
\begin{align}\label{eq:parsing intro}
    x_1^N = w_1^{(N)}(x)\cdots w_{c_N(x)}^{(N)}(x)
\end{align}
into contiguous blocks, whose lengths may vary arbitrarily and may depend on the
underlying sequence $x$.
Such parsings arise naturally in parsing--based
entropy and cross--entropy estimation schemes, including those inspired by the methodology of
Lempel--Ziv \cite{LZ77,LZ78,MZ93}. Parsings also appear in
ergodic--theoretic constructions, where discrete--time trajectories are
partitioned into contiguous blocks according to stopping times or successive
return times; see, for example, \cite{SH96}. Finally, block decompositions of
discrete--time samples are a standard tool in large--deviation and concentration
arguments, where long trajectories are split into contiguous blocks and boundary
effects are controlled under additional dependence assumptions; see, e.g.,
\cite{Rio00}.

Given a parsing as in \eqref{eq:parsing intro}, we consider the quantity
\begin{equation}\label{eq:IN}
     \frac1N \sum_{i=1}^{c_N(x)} -\log \P\bigl([w_i^{(N)}(x)]\bigr),
\end{equation}
and ask under which conditions on the parsing and the underlying measure $\P$ the quantity in~\eqref{eq:IN}
converges $\P$--almost surely and in $\L^1(\P)$ to the entropy rate function $h_\P$.

\medskip

A central difficulty in answering this question is that the quantity~\eqref{eq:IN}
replaces the probability of the entire block \(x_1^N\) by a product of probabilities
of smaller blocks, thereby suppressing all dependence between successive
segments. The SMB theorem, by contrast, concerns the single global quantity
$-\log\P([x_1^N])$, in which correlations are fully retained. The parsing--based
sum~\eqref{eq:IN} therefore treats the blocks as if they were
probabilistically independent.

Moreover, the block boundaries may depend arbitrarily on the observed sample
and on $N$, and may be constructed by any (possibly random)
data--dependent procedure. In particular, the parsing need not arise from
stopping times, nor satisfy any structural compatibility with the left shift.
The blocks could, for instance, be selected so as to minimize or otherwise
bias conditional probabilities across block boundaries.

In the main theorem below (Theorem~\ref{thm:main}), we show that this
convergence question nevertheless admits an affirmative answer at a very high
level of generality. Specifically, the convergence holds whenever the number
of parsing blocks $c_N$ grows sublinearly with the observation length~$N$
and $\P$ is shift--invariant.

The argument does not rely on
large--deviation refinements of the SMB theorem, nor on Gibbsian or explicit
decoupling assumptions controlling conditional probabilities between
successive blocks, as in~\cite{BGPR24,BGPR26}. Instead, it is based on the
martingale decomposition underlying Breiman’s original proof of the SMB
theorem~\cite{Breiman57,Brei60}, and brings out a structural
feature of Breiman’s decomposition—its robustness under arbitrary
sublinear block decompositions—which does not appear to have been
explicitly isolated in this generality.

At this level of generality on the parsing procedure and the underlying measure
$\P$, the sublinear growth condition \eqref{eq:sublinear} below is essentially sharp: if the block
count is allowed to grow linearly, there exist parsings and ergodic measures~$\P$
for which the quantity in~\eqref{eq:IN} fails to converge. This is discussed in
Section~\ref{sec:necessity}.

\medskip

\noindent\textbf{Main result.}
\begin{theorem}\label{thm:main}
Let $\P$ be a shift--invariant probability measure on $(\Omega,\cF)$, and
$\{w_i^{(N)}\}_{i=1}^{c_N}$ any parsing of the form \eqref{eq:parsing intro} satisfying
\begin{align}\label{eq:sublinear}
     \lim_{N\to\infty} \frac{c_N(x)}{N} = 0,
    \qquad \P\text{-a.s.}
\end{align}

Then, with $h_\P$ the entropy rate function appearing in the SMB theorem,
\[
    \lim_{N\to\infty}
    \frac{1}{N}\sum_{i=1}^{c_N(x)} -\log\P\bigl([w_i^{(N)}(x)]\bigr)
    = h_\P(x),
    \qquad \text{holds }\P\text{-a.s.\ and in }\L^1(\P).
\]
In particular, $\E_\P(h_\P)=h(\P)$.
\end{theorem}

Theorem~\ref{thm:main} admits the following structural reformulation as
an approximate factorization of cylinder probabilities under mere
shift--invariance.
\begin{corollary}[Approximate Factorization]\label{cor:factorization}
Let $\P$ be a shift--invariant probability measure on $(\Omega,\cF)$,
and let $x_1^N = w_1^{(N)}\cdots w_{c_N}^{(N)}$ be any parsing of the form
\eqref{eq:parsing intro} satisfying the condition \eqref{eq:sublinear}.
Then there exists a sequence of measurable functions
$r_N:\Omega\to\rr$ satisfying $\frac{r_N(x)}{N}\to 0$, $\P$-a.s and in $\L^1(\P)$, such that
\begin{equation}\label{eq:decoupling0}
\P([x_1^N])
=
\Bigg(\prod_{i=1}^{c_N} \P([w_i^{(N)}])\Bigg)
\exp\{r_N(x)\}.
\end{equation}
\end{corollary}

\begin{remark}
If one relaxes the assumption on the parsing to
$\frac{c_N(x)}{N}\xrightarrow{\P}0$,
then the $\L^1(\P)$ convergence in Theorem~\ref{thm:main} still holds.
\end{remark}

\begin{remark}[Robustness under subextensive perturbations of the parsing]\label{rem:rigid}
A robustness feature of our argument, explored in
Section~\ref{sec:robustness}, is that the stability
Theorem~\ref{thm:main} is insensitive to subextensive perturbations of the
parsing blocks.
More precisely, for every $N$, we can replace each $w_i^{(N)}$ either by
\begin{itemize}
    \item a \emph{sub--block} $\underline w_i^{(N)}\subseteq w_i^{(N)}$, or
    \item a \emph{super--block} $w_i^{(N)}\subseteq \overline w_i^{(N)}$ satisfying
    a non--degeneracy condition, see Section~\ref{sec:robustness},
\end{itemize}
and denote in both cases the resulting block by $\widetilde w_i^{(N)}$.
Assuming that the total modification is subextensive, namely
\begin{equation}\label{eq:subext-mod}
    \sum_{i=1}^{c_N(x)}
    |\widetilde w_i^{(N)}(x)|
    = N+o(N),
    \qquad \P\text{-a.s.},
\end{equation}
the conclusion of Theorem~\ref{thm:main} still holds.
\end{remark}

\subsection{Preview}

The remainder of the paper is organized as follows.
In Section~\ref{sec:main}, we introduce the notation and establish the main
Theorem~\ref{thm:main}. Section~\ref{sec:necessity} establishes the sharpness
of the condition \eqref{eq:sublinear} via a general counterexample.

In Section~\ref{sec:consequences}, we investigate the robustness of the main
result under subextensive perturbations of the parsing, as anticipated in
Remark~\ref{rem:rigid}, and interpret this phenomenon in the spirit of
thermodynamic limits. In Section~\ref{sec:infinitealphabets}, we discuss possible
extensions of Theorem~\ref{thm:main} to non-finite alphabets~$\cA$, within the framework of
Barron’s generalized Shannon--McMillan--Breiman theorem for densities~\cite{Bar85}.

\section{Notation and Proof of Theorem~\ref{thm:main}}\label{sec:main}

\subsection{Notation}\label{sec:notation}

Let the space $\Omega \coloneqq \{(x_k)_{k\in\nn} : x_k \in \cA \text{ for all } k\in\nn\}$ be
equipped with the $\sigma$--field $\cF$ generated by cylinder sets of the form
$[x_1^n] := \{y \in \Omega : y_1^n = x_1^n\}$, where
$x_1^n =( x_1, x_2, \dots, x_n)$ for the sequence $x=(x_k)_{k\geq 1} \in \Omega$.
The \emph{shift} map $T:\Omega \to \Omega$, defined by $(Tx)_k \coloneqq x_{k+1}$,
is a measurable surjection.
We let $\P$ denote a shift--invariant (i.e.\ $T$--invariant) probability measure
on $(\Omega,\cF)$.
For such a measure, we denote by $\P_n$ its $n$--th marginal, that is, the
probability measure on $\cA^n$ defined by
$\P_n(x_1^n) := \P([x_1^n])$. For two positive sequences of real numbers $(a_n)_{n\geq 1}$ and $(b_n)_{n\geq 1}$, we write
\[a_n\sim b_n,\]
whenever
\[\lim_{n\to\infty}\frac{a_n}{b_n}=1.\]
We define
\[
    \supp \P_n := \{a \in \cA^n : \P_n(a) > 0\}
\]
and
\[
    \supp \P := \{x \in \Omega : x_1^n \in \supp \P_n \text{ for all } n \in \nn\}
    = \bigcap_{n\in\nn}\supp\P_n.
\]
For $n \ge 1$, we denote by $H(\P_n)$ the Shannon entropy of the $n$--th marginal,
defined by
\begin{align}\label{eq:shannon entropy}
    H(\P_n) := -\sum_{a \in \cA^n} \P_n(a) \log \P_n(a)\in[0,\log|\cA|].
\end{align}

Finally, the entropy rate of $\P$ is defined as the limit
\begin{align}\label{eq:h}
     h(\P) := \lim_{n\to\infty} \frac{1}{n} H(\P_n),
\end{align}
which exists in $[0,\log|\cA|]$ by sub-additivity.

\subsection{A Martingale Decomposition}\label{subsec:martingale}

As stated in the introduction, the proof of Theorem~\ref{thm:main} relies on a
martingale naturally associated to a shift--invariant probability measure
$\P$ on $(\Omega,\cF,(\cF_n)_{n\geq 1})$, where $\cF_n$ denotes the $\sigma$--field
generated by the cylinder sets $\{[x_1^n] : x_1^n \in \cA^n\}$.
We recall here the basic properties of this martingale and the resulting decomposition of
$-\log \P([x_1^n])$.
The following result is standard, and its proof is postponed to the
appendix (see Section~\ref{sec:proof-lem-martingale}).

\begin{lemma}\label{lem:martingale}
Let $\P$ be a shift-invariant probability measure on $(\Omega,\cF)$.
Define $(Z_n(x))_{n\ge1}$ for $x\in\supp\P$ by
\begin{align}
  Z_n(x):=\frac{\P([x_2^n])}{\P([x_1^n])}, \qquad n\ge2,
\end{align}
and set $Z_1(x):=\P([x_1])^{-1}$. Then:
\begin{enumerate}
\item[(i)]
$(Z_n)_{n\ge1}$ is a
$\big((\cF_n)_{n\ge1},\P\big)$-martingale with values in $[1,\infty)$.
In particular, we have the $\P$-a.s.\ convergence $Z_n\to Z$, thus also $\log Z_n\to\log Z$.

\item[(ii)]
If $Z_{\max}:=\sup_{n\ge1}Z_n$, then
\[
\P(Z_{\max}>t)\le \frac{|\cA|}{t},
\qquad
\P(\log Z_{\max}>t)\le |\cA|\;\Exp{-t},\qquad \forall t>0.
\]
Consequently, $\log Z_{\max}\in\L^1(\P)$ and $\log Z_n\to\log Z$ also holds in $\L^1(\P)$.

\item[(iii)]
With the notation of \eqref{eq:shannon entropy} and \eqref{eq:h}, we have
\[
\E(\log Z_n)=H(\P_n)-H(\P_{n-1}),\quad \E(\log Z)=h(\P).
\]
\end{enumerate}
\end{lemma}

For any $x\in\supp\P$, we observe the identities
\begin{align}\label{eq:decomp0}
     -\log\P([x_1^n])
=\sum_{k=0}^{n-1}\log Z_{n-k}(T^k x)= I_n(x)+J_n(x),
\end{align}
where
\[
I_n(x):=\sum_{k=0}^{n-1}\log Z(T^k x),
\qquad
J_n(x):=\sum_{k=0}^{n-1}\varphi_{n-k}(T^k x),
\quad
\varphi_r:=\log Z_r-\log Z.
\]
The 1957 Breiman proof of the SMB amounts essentially to proving that
$n^{-1}J_n\to0$, both $\P$-a.s.\ and in $\L^1(\P)$
(see~\cite{Breiman57,Brei60}).

We also define, for fixed $M\in\nn$, the truncated and residual quantities
\[
I_{n,M}(x):=\sum_{k=0}^{n-1}\log Z_M(T^k x),
\qquad
J_{n,M}(x):=\sum_{k=0}^{n-1}\varphi_{n-k,M}(T^k x),
\quad
\varphi_{r,M}:=\log Z_r-\log Z_M.
\]
Then
\begin{equation}\label{eq:decomp}
-\log\P([x_1^n]) = I_{n,M}(x)+J_{n,M}(x),
\end{equation}
and the original decomposition \eqref{eq:decomp0} is recovered by letting $M\to\infty$.
This truncation ensures that all relevant quantities depend only on finitely many
coordinates, and will simplify the proof that follows.

\subsection{Parsings}\label{subsec:parsing}

\begin{definition}[Parsing]\label{def:parsing}
A \emph{parsing} of $x_1^N=(x_1,x_2,\dots,x_N)$ is an array of random variables
\[
\big\{(w_i^{(N)}(x))_{i=1}^{c_N(x)}\big\}_{N\ge1},
\]
where $c_N:\Omega\to\nn_{\leq N}$ and
$w_i^{(N)}:\Omega\to\bigcup_{k=1}^N\cA^k$,
such that for every $N\in\nn$ and every $x\in\Omega$,
\[
x_1^N
=
w_1^{(N)}(x)\,w_2^{(N)}(x)\cdots w_{c_N(x)}^{(N)}(x).
\]
We write
\[
\ell_{i,N}:=|w_i^{(N)}|,\qquad
L_{i,N}:=\sum_{j=1}^i \ell_{j,N},
\qquad L_{0,N}:=0.
\]
\end{definition}

\subsection{Three Auxiliary Lemmas}\label{subsec:auxlemmas}

\paragraph{Auxiliary notation.}
Let $(w_i^{(N)})_{i=1}^{c_N}$ be a parsing of $x_1^N=(x_1,x_2,\dots,x_N)$.
For $(m,M)\in\nn^2$ and $k\in\{0,\dots,m\}$, define the random indices
\begin{align}\label{eq:index}
  r_{i,k,N}:=L_{i-1,N}+\ell_{i,N}-1-k,
\end{align}
and introduce the random variables
\begin{align}\label{eq:Ydef}
Y^{(N)}_{i,k,M}(x)
:=
\log Z_{k+1}(T^{r_{i,k,N}}x)
+
\log Z_M(T^{r_{i,k,N}}x)
\end{align}
and
\begin{align}\label{eq:Ydef 0}
Y^{(N)}_{0,k,M}(x)
:=
\log Z_{k+1}(T^{N-k}x)
+
\log Z_M(T^{N-k}x).
\end{align}

We define
\begin{align}\label{eq:res}
    J_{n,M}(w_i^{(N)}(x)):=\sum_{k=0}^{\ell_{i,N}-1}\varphi_{n-k,M}(T^{k+L_{i-1,N}} x),
\end{align}
and suppress the $N$ and $x$ dependence when the context is clear.
Finally, we define
\begin{align}\label{eq:psi}
   \psi_{m,M}
:=
\sup_{k\ge m}\big|\log Z_k-\log Z_M\big|,
\qquad (m,M)\in\nn^2.
\end{align}

\begin{remark}\label{rem:psi}
The non-negative sequence $(\psi_{m,M})_{m,M\geq 1}$ is uniformly bounded by
$2\log Z_{\max}\in\L^1(\P)$ (see Lemma~\ref{lem:martingale}) and is therefore uniformly integrable.
\end{remark}

\begin{lemma}[Birkhoff sum--parsing decomposition]\label{lem:det}
Let $\{(w_i^{(N)}(x))_{i=1}^{c_N(x)}\}_{N\ge1}$ be a parsing.
Then, for every $m\ge M$ and every $N\ge1$, we have the decomposition
\begin{equation}\label{eq:estim}
\sum_{i=1}^{c_N(x)}\big|J_{\ell_i,M}(w_i^{(N)}(x))\big|
\le
\sum_{k=0}^{N-1}\psi_{m,M}(T^k x)
+
\sum_{i=1}^{c_N(x)}\sum_{k=0}^m Y^{(N)}_{i,k,M}(x)
\end{equation}
in terms of the ``Birkhoff sum''
$\sum_{k=0}^{N-1}\psi_{m,M}(T^k x)$
and the
``parsing-dependent'' term
$\sum_{i=1}^{c_N(x)}\sum_{k=0}^m Y^{(N)}_{i,k,M}(x)$.
\end{lemma}

\begin{proof}
Recall that $\varphi_{r,M}=\log Z_r-\log Z_M$. From definition \eqref{eq:res},
\[
J_{\ell_i,M}(w_i^{(N)}(x))
=
\sum_{k=0}^{\ell_i-1}
\varphi_{\ell_i-k,M}(T^{L_{i-1}+k}x),
\]
and therefore
\[
\sum_{i=1}^{c_N}\big|J_{\ell_i,M}(w_i^{(N)}(x))\big|
\le
\sum_{i=1}^{c_N}\sum_{k=0}^{\ell_i-1}
\big|\log Z_{\ell_i-k}(T^{L_{i-1}+k}x)-\log Z_M(T^{L_{i-1}+k}x)\big|.
\]
Reversing the index inside each block yields
\[
\sum_{i=1}^{c_N}\sum_{k=0}^{\ell_i-1}
\big|\log Z_{k+1}(T^{L_{i-1}+\ell_i-1-k}x)-\log Z_M(T^{L_{i-1}+\ell_i-1-k}x)\big|.
\]
Fix $m\ge M$. Splitting the inner sum at $m$ gives, in the notation of \eqref{eq:index},
\begin{align*}
&\sum_{i=1}^{c_N}\sum_{k=m+1}^{\ell_i-1}
\big|\log Z_{k+1}(T^{r_{i,k,N}}x)-\log Z_M(T^{r_{i,k,N}}x)\big|
\\
&\qquad
+\sum_{i=1}^{c_N}\sum_{k=0}^{m}
\big|\log Z_{k+1}(T^{r_{i,k,N}}x)-\log Z_M(T^{r_{i,k,N}}x)\big|,
\end{align*}
By Definition~\ref{def:parsing} and \eqref{eq:psi}, the first term is bounded by
\[
\sum_{i=1}^{c_N}\sum_{k=m+1}^{\ell_i-1}\psi_{m,M}(T^{r_{i,k,N}}x)
\le
\sum_{k=0}^{N-1}\psi_{m,M}(T^k x),
\]
where the inequality holds even when $\ell_i<m$, since $\psi_{m,M}\ge0$.
For the second term, using $|\log Z_{k+1}-\log Z_M|\le \log Z_{k+1}+\log Z_M$, we obtain
\[
\sum_{i=1}^{c_N}\sum_{k=0}^{m}
\big|\log Z_{k+1}(T^{r_{i,k,N}}x)-\log Z_M(T^{r_{i,k,N}}x)\big|
\le
\sum_{i=1}^{c_N}\sum_{k=0}^{m} Y^{(N)}_{i,k,M}(x).
\]
Combining the two bounds yields
\[
\sum_{i=1}^{c_N}\big|J_{\ell_i,M}(w_i^{(N)}(x))\big|
\le
\sum_{k=0}^{N-1}\psi_{m,M}(T^k x)
+
\sum_{i=1}^{c_N}\sum_{k=0}^{m} Y^{(N)}_{i,k,M}(x)
\]
and proves the claim.
\end{proof}

\begin{lemma}[Negligibility of the parsing-dependent term in \eqref{eq:estim}]\label{lem:conc}
Let $\{(w_i^{(N)}(x))_{i=1}^{c_N(x)}\}_{N\ge1}$ be a parsing with $c_N=o(N)$, $\P$-a.s.
Then, for every fixed $m\geq M$,
\begin{align}\label{eq:pars.dep.negl}
  \sum_{i=1}^{c_N(x)}\sum_{k=0}^m Y^{(N)}_{i,k,M}(x)=o(N),\qquad \P\text{-a.s. and in }\L^1(\P).
\end{align}
In particular, the parsing-dependent term in \eqref{eq:estim} of Lemma~\ref{lem:det}
is negligible at scale $N$.
\end{lemma}

\begin{proof}
Fix integers $m\geq M$.
Since $Z_{k+1}$ and $Z_M$ depend on finitely many coordinates and $\cA$ is finite,
we have
\[
\sup_{x\in\supp\P}\log Z_{k+1}(x)<\infty,
\qquad
\sup_{x\in\supp\P}\log Z_M(x)<\infty .
\]
It follows that
\[
\Delta_{m,M}
:=
\sup_{\substack{x\in\supp\P\\ 0\le k\le m}}
\Big(\log Z_{k+1}(x)+\log Z_M(x)\Big)
<\infty,
\]
and $Y^{(N)}_{i,k,M}(x)\leq \Delta_{m,M}$ for every $x\in\supp\P$, $N\in\nn$, $1\leq i\leq c_N(x)$, and $0\leq k\leq m$.
Therefore,
\[
\sum_{i=1}^{c_N(x)}\sum_{k=0}^m Y^{(N)}_{i,k,M}(x)
\le
(m+1)\,\Delta_{m,M}\,c_N(x).
\]
Since $c_N=o(N)$ $\P$-a.s., we obtain
\[
\sum_{i=1}^{c_N(x)}\sum_{k=0}^m Y^{(N)}_{i,k,M}(x)
=o(N),
\qquad
\P\text{-a.s.},
\]
establishing the $\P$-a.s. claim in \eqref{eq:pars.dep.negl}. To deal with the $\L^1(\P)$ claim, we note that by non-negativity and uniform boundedness, the sequence
\begin{align}\label{eq:sequ}
  \Bigg(
\frac{1}{N}\sum_{i=1}^{c_N(x)}\sum_{k=0}^m Y^{(N)}_{i,k,M}(x)
\Bigg)_{N\ge 1}
\end{align}
is uniformly integrable. Since it converges to zero \(\P\)-a.s., the convergence also holds in \(\L^1(\P)\).
\end{proof}

\begin{remark}\label{rem:in prob}
Lemma~\ref{lem:conc} is the only point in the argument where the sublinearity
assumption~\eqref{eq:sublinear} is used. Because the sequence~\eqref{eq:sequ} is
uniformly integrable, convergence to zero in \(\L^1(\P)\) already follows from the
weaker condition
\[
\frac{c_N}{N} \xrightarrow{\P} 0.
\]
This condition ensures that, for every fixed $m \geq M$, the sequence~\eqref{eq:sequ}
converges to zero in probability as $N \to \infty$.
\end{remark}

\begin{lemma}[Negligibility of the Birkhoff sum in \eqref{eq:estim}]\label{lem:conc.ergo}
There exist sequences of natural numbers $(\widetilde M_j)_{j\geq 0}$ and non-negative functions $(\widehat{\psi}_{\infty,\widetilde M_j})_{j\geq 1}$, such that
\begin{align}\label{eq:erg.a.s.}
 \lim_{m\to\infty}
\lim_{N\to\infty}
\frac{1}{N}\sum_{k=0}^{N-1}\psi_{m,\widetilde M_j}(T^k x)=\widehat{\psi}_{\infty,\widetilde M_j},\qquad
\lim_{j\to\infty}\widehat{\psi}_{\infty,\widetilde M_j}=0,\qquad
\text{hold }\P\text{-a.s.}
\end{align}
Moreover,
\begin{align}\label{eq:ergexp}
  \lim_{M\to\infty}\lim_{m\to\infty}
\lim_{N\to\infty}
\E\!\left[\frac{1}{N}\sum_{k=0}^{N-1}\psi_{m,M}(T^k )\right]=0.
\end{align}
\end{lemma}

\begin{proof}
Fix $m,M\in\mathbb N$.
By the pointwise ergodic theorem for shift-invariant measures~\cite[Theorem~I.3.1]{SH96} and Remark~\ref{rem:psi}, there exists a
shift-invariant function $\widehat\psi_{m,M}\in \L^1(\P)$ such that
\[
\lim_{N\to\infty}
\frac{1}{N}\sum_{k=0}^{N-1}\psi_{m,M}(T^k)=\widehat\psi_{m,M},
\qquad \P\text{-a.s.},
\]
and $\E(\widehat\psi_{m,M})=\E(\psi_{m,M})$. Since $\log Z_{\max} \in \L^1(\P)$, this same result applies
to $\log Z_{\max}$ as well. Consequently, there exists a shift-invariant
$\Psi \in \L^1(\P)$ such that
\[
\lim_{N\to\infty}
\frac{1}{N}\sum_{k=0}^{N-1}\log Z_{\max}(T^k)=\Psi,
\qquad \P\text{-a.s.},
\]
with $\E(\Psi)=\E(\log Z_{\max})$.

Since the sequence $(\psi_{m,M})_{m\ge1}$ of \eqref{eq:psi} is non-increasing, the same holds for
$(\widehat\psi_{m,M})_{m\ge1}$, and moreover $0\leq\widehat\psi_{m,M}\le 2\Psi$ for all $m,M$.
Hence $\widehat\psi_{m,M}\downarrow \widehat\psi_{\infty,M}$ $\P$-a.s. as $m\to\infty$.
Using
\[
\lim_{m\to\infty}\psi_{m,M}=|\log Z-\log Z_M|,
\qquad
\E(\psi_{1,M})\le 2\E(\log Z_{\max})<\infty,
\]
the monotone convergence theorem yields
\begin{equation}\label{eq:psi-limit}
\E(\widehat\psi_{\infty,M})
=\lim_{m\to\infty}\E(\psi_{m,M})
=\E(|\log Z-\log Z_M|);
\end{equation}
and by Lemma~\ref{lem:martingale} again,
\begin{equation}\label{eq:L1-conv}
\lim_{M\to\infty}\E(|\log Z-\log Z_M|)=0.
\end{equation}
From \eqref{eq:psi-limit}--\eqref{eq:L1-conv}, we obtain $\widehat\psi_{\infty,M}\to 0$ in probability as $M\to\infty$; therefore, there exists a sequence of natural numbers $(\widetilde M_j)_{j\geq 0}$ such that
\begin{align}
    \lim_{j\to\infty}\widehat\psi_{\infty,\widetilde M_j}=0,\qquad \P\text{-a.s.},
\end{align}
establishing \eqref{eq:erg.a.s.}.
Using shift-invariance and linearity,
\[
\E\!\left[\frac{1}{N}\sum_{k=0}^{N-1}\psi_{m,M}(T^k )\right]=\E(\psi_{m,M}).
\]
The convergence in expectation follows from \eqref{eq:psi-limit}--\eqref{eq:L1-conv}.
\end{proof}

\subsection{Proof of Theorem~\ref{thm:main}}

\begin{proof}
Applying the decomposition~\eqref{eq:decomp} to each block $(w_i^{(N)}(x))_{i=1}^{c_N}$, where
\[
I_{\ell_i,M}(w_i^{(N)}(x)):=\sum_{k=0}^{\ell_i-1}\log Z_M(x_{L_{i-1}+1}^{L_{i-1}+\ell_i}),
\]
and summing over $i$, yields
\begin{align}\label{eq:block-decomp}
\sum_{i=1}^{c_N}-\log\P([w_i^{(N)}])
&=\sum_{i=1}^{c_N}\Big(I_{\ell_i,M}(w_i^{(N)})+J_{\ell_i,M}(w_i^{(N)})\Big) \\
&=\sum_{i=1}^{c_N} I_{\ell_i,M}(x_{L_{i-1}+1}^{L_{i-1}+\ell_i})
 +\sum_{i=1}^{c_N} J_{\ell_i,M}(x_{L_{i-1}+1}^{L_{i-1}+\ell_i}).
\end{align}
The first term telescopes:
\begin{align}
\sum_{i=1}^{c_N} I_{\ell_i,M}(x_{L_{i-1}+1}^{L_{i-1}+\ell_i})
&=\sum_{i=1}^{c_N}\sum_{k=0}^{\ell_i-1}\log Z_M(T^{L_{i-1}+k}x)\label{eq:sum of blocks}\\
&=\sum_{j=0}^{N-1}\log Z_M(T^j x)\\
&= I_{N,M}(x)\\
&=-\log\P([x_1^N]) - J_{N,M}(x),
\end{align}
where the last identity follows from applying the decomposition~\eqref{eq:decomp}
to the block $x_1^N$.

Combining the above gives
\begin{equation}\label{eq:final-bound}
\Bigg|
\frac{1}{N}\sum_{i=1}^{c_N}-\log\P([w_i^{(N)}])
+\frac{1}{N}\log\P([x_1^N])
\Bigg|
\le
\frac{1}{N}|J_{N,M}(x)|
+\frac{1}{N}\sum_{i=1}^{c_N}|J_{\ell_i,M}(w_i^{(N)})|,\qquad \forall M\in\nn.
\end{equation}

We note that the term $J_{N,M}(x)$ corresponds to the special case of the trivial
parsing $w_1^{(N)}=x_1^N$. We first prove the $\P$-a.s.\ convergence to zero of the left-hand side of \eqref{eq:final-bound}. Observe that \eqref{eq:final-bound} remains valid when replacing $M$ by each term of the subsequence $(\widetilde M_j)_{j\geq 0}$ from Lemma~\ref{lem:conc.ergo}.
Referring to the notation of \eqref{eq:Ydef}, \eqref{eq:Ydef 0}, and \eqref{eq:psi}, Lemma~\ref{lem:det} shows that the right-hand side of \eqref{eq:final-bound} is bounded from above by
\begin{align}\label{eq:finalfr}
    \frac{2}{N}\sum_{k=0}^{N-1}\psi_{m, \widetilde M_j}(T^kx)+\frac{1}{N}\sum_{k=0}^mY^{(N)}_{0,k,\widetilde M_j}(x)+\frac{1}{N}\sum_{i=1}^{c_N(x)}\sum_{k=0}^m Y^{(N)}_{i,k,\widetilde M_j}(x),\qquad \forall m\geq \widetilde M_j.
\end{align}
Equations \eqref{eq:pars.dep.negl}
and~\eqref{eq:erg.a.s.} imply that~\eqref{eq:finalfr}
converges $\P$-almost surely to
$2\widehat{\psi}_{\infty,\widetilde M_j}(x)$,
after sending $N\to\infty$ and then $m\to\infty$.

Finally, \eqref{eq:erg.a.s.} gives the $\P$-almost sure convergence to zero of $\widehat{\psi}_{\infty,\widetilde M_j}$ as $j\to\infty$. Therefore, the left-hand side of \eqref{eq:final-bound} converges to zero $\P$-a.s.

Convergence to zero in $\L^1(\P)$ of the left-hand side of \eqref{eq:final-bound} is obtained in the same way. Indeed, it follows from \eqref{eq:finalfr} combined with \eqref{eq:pars.dep.negl} and \eqref{eq:ergexp}.

The conclusion follows by applying the Shannon--McMillan--Breiman theorem to
$-\frac{1}{N}\log\P([x_1^N])$.
\end{proof}

\subsection{The Necessity of Sublinear Block Count}
\label{sec:necessity}

As stated in the introduction, the condition \eqref{eq:sublinear} is essential for convergence at this level of generality.
Indeed, for any $K\in\mathbb{N}$, one can construct parsings
$\{w^{(N)}_i\}_{i=1}^{c_N}$ satisfying
\begin{align}\label{eq:extensive}
    c_N \sim \frac{N}{K}, \qquad \P\text{-a.s.},
\end{align}
such that
\begin{align}\label{eq:example nonexistence}
 \lim_{N\to\infty}
\frac{1}{N}\sum_{i=1}^{c_N} -\log\P([w_i^{(N)}])
\quad\text{does not exist }\P\text{-a.s.},
\end{align}
for any $\P$ that is not a Markov measure of order $\leq K$.

\medskip

Let $K\in\mathbb{N}$ be even (this assumption is only for the sake of
the construction below), and let $\P$ be an ergodic measure such that
\begin{align}\label{eq:discrepency}
   \frac{1}{K}H(\P_K)
>
\frac{1}{2}\Big(\frac{2}{K}H(\P_{\frac{K}{2}})+h(\P)\Big).
\end{align}
The inequality \eqref{eq:discrepency} holds if $\P$ is not a Markov measure of order
$\le K$.
To see this, write
\[
H(\P_K)=H(\P_{K/2})+\sum_{n=K/2+1}^{K}\beta_n(\P),
\qquad
\beta_n(\P):=H(\P_n)-H(\P_{n-1}).
\]
It is well-known that the sequence $(\beta_n(\P))_{n\ge1}$ is non-increasing, converges to
$h(\P)$, and $\beta_{m+1}(\P)=h(\P)$ if and only if $\P$ is a Markov measure of
order $\leq m$; see~\cite[Ch.~4]{CT06}. Hence, inequality~\eqref{eq:discrepency} holds if and only if there exists
$n\in\{K/2+1,\dots,K\}$ such that $\beta_n(\P)>h(\P)$. Therefore, inequality~\eqref{eq:discrepency} holds for any $\P$ which is not a Markov measure of order $\leq K$.

\medskip

We will construct two distinct parsings
$\{u^{(N)}_i\}_{i=1}^{a_N}$ and
$\{v^{(N)}_i\}_{i=1}^{b_N}$, both satisfying
\eqref{eq:extensive}, and such that
\begin{itemize}
    \item[a.]
    The parsing $\{u^{(N)}_i\}_{i=1}^{a_N}$ is deterministic (that is, does not depend on the observed data) and
    \begin{align}\label{eq:deterministicfinalexam}
      \lim_{N\to\infty}
    \frac{1}{N}\sum_{i=1}^{a_N}-\log\P([u_i^{(N)}])
    =
    \frac{1}{K}H(\P_K),
    \qquad
    \P\text{-a.s.}
    \end{align}

    \item[b.]
    The parsing $\{v^{(N)}_i\}_{i=1}^{b_N}$ is data--dependent and
   \begin{align}\label{eq:datadependentexample}
      \lim_{N\to\infty}
    \frac{1}{N}\sum_{i=1}^{b_N}-\log\P([v_i^{(N)}])
    =
    \frac{1}{2}\Big(\frac{2}{K}H(\P_{\frac{K}{2}})+h(\P)\Big),
    \qquad
    \P\text{-a.s.}
   \end{align}
\end{itemize}
We conclude by defining
\[
\{w^{(N)}_i\}_{i=1}^{c_N}
:=
\begin{cases}
\{u^{(N)}_i\}_{i=1}^{a_N}, & \text{if $N$ is even},\\
\{v^{(N)}_i\}_{i=1}^{b_N}, & \text{if $N$ is odd},
\end{cases}
\]
and noting that we have \eqref{eq:example nonexistence} for it.

\subsubsection{The Construction of $\{u^{(N)}_i\}_{i=1}^{a_N}$}

We set
\[
u_1^{(N)}(x):=x_1^{K},\quad
u_2^{(N)}(x):=x_{K+1}^{2K},\ \ldots,\
u_{a_N-1}^{(N)}(x):=
x_{K(\lfloor\frac{N}{K}\rfloor-1)+1}^{K\lfloor\frac{N}{K}\rfloor},
\]
and
\[
u_{a_N}^{(N)}(x):=
x_{K\lfloor\frac{N}{K}\rfloor+1}^{N}.
\]
Then $a_N\sim \frac{N}{K}$ and
\begin{align}\label{eq:exam}
\frac{1}{N}\sum_{i=1}^{a_N}-\log\P([u_i^{(N)}])
&=
\frac{1}{K}\sum_{j=1}^{K}
\frac{K}{N}
\sum_{n=0}^{\lfloor\frac{N}{K}\rfloor-1}
\log Z_j(T^{Kn+j}x)+o(1),
\end{align}
where the $o(1)$ term comes from the last $\leq K$ symbols coming from $u_{a_N}^{(N)}(x)$. This term is negligible compared to $N$ because the random variables $Z_j$
(for $j\le K$) are uniformly bounded, as a direct consequence of the finiteness of the
alphabet $\cA$.

By the Birkhoff pointwise ergodic theorem~\cite[Theorem~I.3.1]{SH96} and Lemma~\ref{lem:martingale},
for every $j=1,\dots,K$,
\begin{align}\label{eq:exam0}
\lim_{N\to\infty}
\frac{K}{N}
\sum_{n=0}^{\lfloor\frac{N}{K}\rfloor-1}
\log Z_j(T^{Kn+j}x)
=
H(\P_j)-H(\P_{j-1}),
\qquad
\P\text{-a.s.}
\end{align}
where $H(\P_0):=0$.
Combining \eqref{eq:exam} and \eqref{eq:exam0}, we obtain
\[
\lim_{N\to\infty}
\frac{1}{N}\sum_{i=1}^{a_N}-\log\P([u_i^{(N)}])
=
\frac{1}{K}H(\P_K),
\qquad
\P\text{-a.s.}
\]

\subsubsection{The Construction of $\{v^{(N)}_i\}_{i=1}^{b_N}$}

We fix $0<\epsilon\ll1$, and consider a set $E_\epsilon\in\mathcal{F}$ such
that $\P(E_\epsilon)\ge 1-\frac{1}{2}\epsilon$ and such that, uniformly on
$E_\epsilon$,
\[
\lim_{n\to\infty}
\frac{1}{n}\log\P([x_1^n])
=
h(\P).
\]
The existence of such a set follows from Egoroff’s theorem \cite{Rud87} applied to the
Shannon--McMillan--Breiman theorem. Furthermore, by the Birkhoff pointwise ergodic theorem, there exists
a full $\P$--measure set $\Omega_{E_\epsilon}$, such that
\[
\lim_{N\to\infty}
\frac{1}{N}\sum_{k=0}^{N-1}\mathds{1}_{E_\epsilon}(T^k x)
=
\P(E_\epsilon),\qquad \text{everywhere on }\Omega_{E_\epsilon}.
\]

For every $x\in\Omega_{E_\epsilon}$ and $N$ sufficiently large, there
exists $k(x)\in[(\frac12-\epsilon)N,\frac12 N]$ such that
$T^{k(x)}x\in E_\epsilon$.
We then define
\[
v_1^{(N,\epsilon)}:=x_1^{\frac{K}{2}},\quad
v_2^{(N,\epsilon)}:=x_{\frac{K}{2}+1}^{K},\ \ldots,\
v_{\lfloor\frac{k(x)}{K}\rfloor}^{(N,\epsilon)}:=
x_{k(x)-\frac{K}{2}}^{k(x)-1},
\]
and
\[
v_{\lfloor\frac{k(x)}{K}\rfloor+1}^{(N,\epsilon)}:=x_{k(x)}^{N}.
\]
Setting $b_{N,\epsilon}(x):=\lfloor\frac{k(x)}{K}\rfloor+1$, we conclude
\[
\frac{N(1-2\epsilon)}{K}
\le b_{N,\epsilon}
\le \frac{N}{K}
\quad\text{eventually }\P\text{-a.s.},
\]
and an argument similar to that of part~[a.] yields a full--measure set
$\Omega_\epsilon$, such that
\[
\limsup_{N\to\infty}
\frac{1}{N}\sum_{i=1}^{b_{N,\epsilon}}
-\log\P([v_i^{(N,\epsilon)}])
\le
\frac{1}{2}\Big(\frac{2}{K}H(\P_{\frac{K}{2}})+h(\P)\Big),
\]
\[
\liminf_{N\to\infty}
\frac{1}{N}\sum_{i=1}^{b_{N,\epsilon}}
-\log\P([v_i^{(N,\epsilon)}])
\ge
\frac{1}{2}\Big(
\frac{2(1-2\epsilon)}{K}H(\P_{\frac{K}{2}})
+
(1+2\epsilon)h(\P)
\Big),
\]
and
\[
\frac{1-2\epsilon}{K}
\le
\liminf_N\frac{b_{N,\epsilon}}{N}
\le
\limsup_N\frac{b_{N,\epsilon}}{N}
\le
\frac{1}{K},
\]
hold everywhere on $\Omega_\epsilon$. Now consider a nonincreasing sequence
$\epsilon_j\in(0,\frac14)$ with $\epsilon_j\downarrow0$, and define
$\Omega^*:=\bigcap_{j\ge1}\Omega_{\epsilon_j}$.
For every $x\in\Omega^*$, there exists a sequence of integers
$(N_j(x))_{j\ge0}$ such that for all $N\ge N_j(x)$,
\[
\frac{1}{N}
\sum_{i=1}^{b_{N,\epsilon_j}(x)}
-\log\P([v_i^{(N,\epsilon_j)}(x)])
\ge
\frac{1}{2}\Big(
\frac{2(1-4\epsilon_j)}{K}H(\P_{\frac{K}{2}})
+
(1+\epsilon_j)h(\P)
\Big),
\]
and
\[
\frac{N(1-4\epsilon_j)}{K}
\le b_{N,\epsilon_j}(x)
\le \frac{N}{K}.
\]

Define the parsing $\{v^{(N)}_i\}_{i=1}^{b_N}$ on $\Omega^*$ by
\[
\{v^{(N)}_i(x)\}_{i=1}^{b_N(x)}
:=
\{v^{(N,\epsilon_j)}_i(x)\}_{i=1}^{b_{N,\epsilon_j}(x)},
\quad
N\in[N_j(x),N_{j+1}(x)).
\]
Then
\[
b_N(x)\sim\frac{N}{K},
\]
and
\[
\liminf_{N\to\infty}
\frac{1}{N}\sum_{i=1}^{b_N(x)}
-\log\P([v_i^{(N)}(x)])
\ge
\frac{1}{2}\Big(\frac{2}{K}H(\P_{\frac{K}{2}})+h(\P)\Big),
\]
\[
\limsup_{N\to\infty}
\frac{1}{N}\sum_{i=1}^{b_N(x)}
-\log\P([v_i^{(N)}(x)])
\le
\frac{1}{2}\Big(\frac{2}{K}H(\P_{\frac{K}{2}})+h(\P)\Big),
\]
hold everywhere on $\Omega^*$, which proves part~[b.].

\section{Robustness, Interpretation, and Extensions}\label{sec:consequences}

\subsection{Robustness Under Subextensive Perturbations}
\label{sec:robustness}

We noted in Remark~\ref{rem:rigid} that the argument underlying
Theorem~\ref{thm:main} is robust under subextensive perturbations of the
parsing blocks. We now make this statement precise.

For each $N\ge1$, let
\[
\bigl\{ w_i^{(N)}(x) \bigr\}_{i=1}^{c_N(x)}
\]
denote the blocks of the parsing of $x_1^N$, and define the corresponding block
lengths and starting indices by
\[
\ell_{i,N}(x) := |w_i^{(N)}(x)|,
\qquad
L_{i-1,N}(x) := \sum_{j=1}^{i-1} \ell_{j,N}(x).
\]

\medskip
We consider two types of perturbations of the blocks
$\bigl\{ w_i^{(N)}(x) \bigr\}_{i=1}^{c_N(x)}$.

\medskip
\noindent\textbf{Sub--block perturbations.}
The block $w_i^{(N)}$ is replaced by a sub--block
\(
\underline w_i^{(N)} \subseteq w_i^{(N)}.
\)

\medskip
\noindent\textbf{Super--block perturbations.}
The block $w_i^{(N)}$ is replaced by a super--block
\(
\overline w_i^{(N)} \supseteq  w_i^{(N)}.
\)
Writing $\overline \ell_{i,N} := |\overline w_i^{(N)}|$ and denoting by
$\overline L_{i-1,N}$ the starting index of the super-block $\overline w_i^{(N)}$, we assume that
the interval
\[
\bigl[\overline L_{i-1,N},\, \overline L_{i-1,N} + \overline \ell_{i,N}\bigr)
\]
has non--empty intersection only with the neighboring original blocks\footnotemark,
namely
\[
\bigl[ L_{i-2,N},\, L_{i-2,N} + \ell_{i-1,N}\bigr)
\quad\text{and}\quad
\bigl[ L_{i,N},\, L_{i,N} + \ell_{i+1,N}\bigr).
\]
\footnotetext{The argument also extends to the case where the intervals
\(
[\overline L_{i-1,N},\, \overline L_{i-1,N} + \overline \ell_{i,N})
\)
are allowed to intersect a uniformly bounded number of intervals of the form
\(
[ L_{i,N},\, L_{i,N} + \ell_{i+1,N})
\).}

For each $N$, the choice of perturbation scheme is allowed to vary from block to
block. Denoting by
\(
\{\widetilde w_i^{(N)}(x)\}_{i=1}^{c_N(x)}
\)
the resulting modified blocks, we assume that the total modification is
subextensive, in the sense that
\[
\sum_{i=1}^{c_N(x)}
|\widetilde w_i^{(N)}(x)|
= N+o(N),
\qquad \P\text{-a.s.}
\]
\textit{Under these perturbations, the conclusion of Theorem~\ref{thm:main} remains
valid.}

\begin{proof}
We apply the same block decomposition as in
\eqref{eq:block-decomp} and~\eqref{eq:sum of blocks}, now to the modified blocks
$\widetilde w_i^{(N)}$, with lengths
\(
\widetilde \ell_{i,N} := |\widetilde w_i^{(N)}|
\)
and corresponding starting indices $\widetilde L_{i,N}$.
The bound~\eqref{eq:final-bound} is replaced by the $\P$-a.s. inequality
\begin{align}\label{eq:finish2.0}
  \Bigg|
\sum_{i=1}^{c_N}
    -\log \P([\widetilde w_i^{(N)}])
+\log \P([x_1^N])
\Bigg|
\leq
|J_{N,M}(x)|
+ o(N)\,\Delta_M
+ \sum_{i=1}^{c_N}
    \bigl|J_{\widetilde \ell_{i,N},M}(\widetilde w_i^{(N)})\bigr|,
\end{align}
where
\[
\Delta_M := \sup_{y_1^M\in\cA^M} \log Z_M(y_1^M).
\]

The first term on the right-hand side of \eqref{eq:finish2.0} is handled exactly as in the proof
of Theorem~\ref{thm:main}. Furthermore, the subextensive assumption ensures
that the additional contribution proportional to $\Delta_M$ vanishes after
normalization by $N$.

The remaining error term
\[
\sum_{i=1}^{c_N}
\bigl|J_{\widetilde \ell_{i,N},M}(\widetilde w_i^{(N)})\bigr|
\]
is controlled by the same auxiliary estimates established in
Section~\ref{subsec:auxlemmas}. The conclusion then follows by the same limiting
procedure as in the proof of Theorem~\ref{thm:main}, completing the argument.
\end{proof}

\begin{remark}
We do not know whether an analogue of Theorem~\ref{thm:main} continues to hold in
the absence of a condition controlling the number of overlaps between
super--blocks and neighboring parsed blocks. In other words, we do not know if the conclusion of Theorem~\ref{thm:main} still holds if the only requirement imposed on sub-block and super-block perturbation is
\[
\sum_{i=1}^{c_N(x)}|\widetilde w_i(x)|=N+o(N),\qquad\P\text{-a.s.}
\]
\end{remark}

\subsection{Coarse--graining}
\label{sec:coarsegraining}

In this section, we draw a parallel between Theorem~\ref{thm:main}, the
robustness results of Section~\ref{sec:robustness}, and the notion of
\emph{coarse--graining} familiar from statistical mechanics.

From an information--theoretic viewpoint, a length--$N$ cylinder
$[x_1^N]$ may be regarded as a microscopic configuration, whose information
content is measured by $-\log\P([x_1^N])$. The Shannon--McMillan--Breiman theorem
asserts that, for $\P$--almost every realization,
\[
-\frac{1}{N}\log\P([x_1^N]) \longrightarrow h(\P),
\]
so that the specific information content of a typical microscopic configuration
converges to a deterministic limit $h(\P)$ whenever $\P$ is ergodic.

Coarse--graining, in a broad sense, consists of replacing a detailed microscopic
description by a collection of aggregated observables at a ``mesoscopic'' scale defined on larger
``blocks,'' and studying the corresponding normalized quantities. In our
setting, a parsing $x_1^N = w_1^{(N)}\cdots w_{c_N}^{(N)}$ with $c_N=o(N)$ induces
such a coarse description, where the relevant observables are the block information contents
$-\log\P([w_i^{(N)}])$.

Theorem~\ref{thm:main} shows that, despite this coarser description, the
normalized sum of blockwise information contents
\[
\frac{1}{N}\sum_{i=1}^{c_N}-\log\P([w_i^{(N)}])
\]
has the same almost--sure limit as the microscopic quantity
$-\frac{1}{N}\log\P([x_1^N])$. By contrast, Section~\ref{sec:necessity} shows that
allowing $c_N$ to grow faster than sublinearly may destroy the existence of this
limit. Finally, the robustness results of Section~\ref{sec:robustness} reinforce
this coarse--graining analogy, by showing that subextensive modifications of the
parsing do not affect the asymptotic specific information content.

\subsection{Extension to Non-finite Alphabets}\label{sec:infinitealphabets}

In this section we briefly indicate how Theorem~\ref{thm:main}
can potentially be extended to the setting of Barron’s generalized
Shannon--McMillan--Breiman theorem for densities~\cite{Bar85}.

The principal technical ingredient in the proof of
Theorem~\ref{thm:main} is Lemma~\ref{lem:martingale},
which relies on the finiteness of the alphabet $\cA$.
Barron showed in his seminal work~\cite{Bar85} that an analogue of
Lemma~\ref{lem:martingale} continues to hold in the setting of densities,
where conditional probabilities are replaced by Radon--Nikodym
derivatives with respect to a reference Markov measure.
Thus, the martingale component of the argument admits a natural
extension beyond finite alphabets.

Apart from Lemma~\ref{lem:martingale}, finiteness of $\cA$
enters the proof through Lemma~\ref{lem:conc}, where it is used to
control the boundary error terms
\[
\sum_{k=0}^m Y^{(N)}_{i,k,M},
\]
arising from the arbitrary parsing.
More precisely, finiteness of $\cA$ ensures that these boundary
contributions are negligible, in the sense that for every fixed $m\ge M$,
\[
  \sum_{i=1}^{c_N(x)} \sum_{k=0}^m Y^{(N)}_{i,k,M}(x)
  = o(N),
  \qquad \P\text{-a.s.\ and in }\L^1(\P).
\]
In the finite--alphabet setting, this ultimately follows from a
uniform boundedness property of the local error terms, which
prevents the random boundary effects induced by the parsing from
accumulating at linear scale.

The finiteness assumption may, however, be bypassed by appealing
directly to an ergodic--theoretic argument. The following consequence
of the ergodic theorem provides the required control; its proof is
given in Appendix~\ref{sec:proof-prop-extension}.

\begin{proposition}\label{prop:extension}
Let $T:\Omega\to\Omega$ be a measure--preserving transformation
of the probability space $(\Omega,\cF,\P)$, and let
$g\in\L^1(\P)$. Let $A_N:\Omega\to 2^{\{1,\dots,N\}}$
be a sequence of set--valued functions satisfying
\[
\lim_{N\to\infty}\frac{|A_N(x)|}{N}=0,
\qquad \P\text{-a.s.}
\]
Then
\[
\lim_{N\to\infty}
\frac{1}{N}\sum_{k\in A_N(x)} g(T^k x)
=0,
\qquad \P\text{-a.s.\ and in }\L^1(\P).
\]
\end{proposition}

Proposition~\ref{prop:extension} shows that sublinear collections of indices
cannot contribute at linear scale to ergodic averages of
integrable functions. Consequently, once the martingale
decomposition is available in the density setting,
there appears to be no essential obstruction to extending
Theorem~\ref{thm:main} by replacing the information contents
$-\log \P([w_i^{(N)}])$ with the corresponding density
terms arising in Barron’s generalized SMB theorem~\cite{Bar85}.

We do not pursue this generalization here.
The finite--alphabet framework is conceptually and technically simpler, and admits a clean
``approximate factorization'' interpretation~\eqref{eq:decoupling0}.

\appendix
\section{Appendix}\label{appendix}

\subsection{Proof of Lemma~\ref{lem:martingale}}\label{sec:proof-lem-martingale}

\begin{proof}
Observe that, given the first bound in (ii), the second is trivial, and then the rest of the claim is a simple application of the representation
\[
\mathbb{E}_\P(X)=\int_0^\infty\P(X>t)\;dt,
\]
valid for any measurable $X\geq0$, and the Dominated Convergence Theorem. Also, the $\P$-a.s.\ existence of the limit $Z$ is guaranteed by Doob's Martingale Convergence Theorem, because $Z_n^{-}=0$ for every $n$ and so $\sup_{n\geq2}\mathbb{E}_\P(Z_n^-)=0<\infty$.

To show that $Z_n$ is a martingale, it suffices to prove
\[
\frac{1}{\P([x_1^n])}\mathbb{E}_{\P}(Z_{n+1}\mathds{1}_{[x_1^n]})= Z_n(x_1^n),
\]
for any $x_1^n\in\supp\P_n$. To this end, note that
\begin{align*}
\mathbb{E}_{\P}(Z_{n+1}\mathds{1}_{[x_1^n]})
&=\sum_{a\in\cA,\;\P([x_1^na])>0}Z_{n+1}(x_1^na)\P([x_1^na])\\
&=\sum_{a\in\cA,\;\P([x_1^na])>0}\frac{\P([x_2^na])}{\P([x_1^na])}\P([x_1^na])\\
&=\sum_{a\in\cA,\;\P([x_1^na])>0}\P([x_2^na])\\
&=\sum_{a\in\cA}\P([x_2^na])\\
&=\P([x_2^n])=\P([x_1^n])Z_n(x_1^n),
\end{align*}
which proves (i.). To prove~(iii), we adapt the argument of Ionescu--Tulcea~\cite{Ion61}
to the present setting.
Fix $t>0$ and consider the set $A:=\{Z_{\text{max}}>t\}$ and the disjoint sets
$A_k:=\{Z_1\leq t,\dots,Z_{k-1}\leq t, Z_k>t\}\in\cF_k$ with $\bigcup_{k\geq 1}A_k=A$.
Applying Markov's inequality to $\P(A_k)$ gives
\[
\P(A)=\sum_{k\geq1}\P(A_k)\leq\frac{1}{t}\sum_{k\geq1}\int_{A_k}Z_k\;\text{d}\P
\leq\frac{1}{t}\sum_{k\geq1}\sum_{x_1^k\in A_k}\P([x_2^k]).
\]
Now consider the unique probability measure $\nu$ on $(\Omega,\cF)$ with marginals $\nu_n$ defined by
\[
\nu_n(x_1,\dots,x_n)=\frac{1}{|\cA|}\P([x_2^n]), \qquad n\in\nn
\]
(this is a standard application of the Daniell--Kolmogorov Consistency Theorem).
The above estimate becomes
\[
\P(A)\leq \frac{|\cA|}{t}\sum_{k\geq 2}\nu_k(A_k)=\frac{|\cA|}{t}\nu(A)\leq \frac{|\cA|}{t}.
\]
Here we have used the fact that the sets $A_k\in\cF_k$, $k\in\nn$ are disjoint, and that $\nu$ is a probability measure on $(\Omega,\cF)$. This proves (ii). The first part of (iii) is a simple computation. The second part follows from an application of (ii), the Cesàro Lemma, and the fact that for any shift-invariant $\P$,
\[
\lim_{n\to\infty}\frac{H(\P_n)}{n}=h(\P).
\]
\end{proof}

\subsection{Proof of Proposition~\ref{prop:extension}}\label{sec:proof-prop-extension}

\begin{proof}
Fix $M\in\mathbb{N}$ and decompose
\begin{equation}\label{eq:decomp-tail}
g = g_M + h_M,
\qquad
g_M := g\,\mathbf{1}_{\{|g|\le M\}},
\quad
h_M := g\,\mathbf{1}_{\{|g|>M\}}.
\end{equation}
Since $|g_M|\le M$, we have
\begin{equation}\label{eq:bounded}
\frac{1}{N}\sum_{k\in A_N(x)} |g_M(T^k x)|
\le
M\,\frac{|A_N(x)|}{N}.
\end{equation}
By assumption,
\begin{equation}\label{eq:sublinear2}
\frac{|A_N(x)|}{N} \to 0,
\qquad \P\text{-a.s.,}
\end{equation}
and hence the term in \eqref{eq:bounded} converges to $0$
$\P$-a.s.\ and in $\L^1(\P)$.

For the tail part,
\begin{equation}\label{eq:tail-bound}
\frac{1}{N}\sum_{k\in A_N(x)} |h_M(T^k x)|
\le
\frac{1}{N}\sum_{k=0}^{N-1} |h_M(T^k x)|.
\end{equation}
By the ergodic theorem,
\begin{equation}\label{eq:ergodic}
\frac{1}{N}\sum_{k=0}^{N-1} |h_M(T^k x)|
\longrightarrow
\E\big[|h_M| \mid \mathcal{T}\big],
\qquad \P\text{-a.s.\ and in }\L^1(\P),
\end{equation}
where $\mathcal{T}$ denotes the $T$-invariant $\sigma$-field.

Since $|h_M| \le g\in\L^1(\P)$ and $h_M \to 0$ $\P$-a.s., we have
\begin{equation}\label{eq:cond}
\lim_{M\to\infty}\E\big[|h_M| \mid \mathcal{T}\big]= 0,
\qquad \P\text{-a.s.\ and in }\L^1(\P).
\end{equation}
Finally,
\[
\frac{1}{N}\Big|\sum_{k\in A_N(x)} g(T^k x)\Big|
\le
M\,\frac{|A_N(x)|}{N}
+
\frac{1}{N}\sum_{k=0}^{N-1} h_M(T^k x).
\]
Letting $N\to\infty$ and using \eqref{eq:sublinear2} and \eqref{eq:ergodic},
then letting $M\to\infty$ and applying \eqref{eq:cond}, yields the result.
\end{proof}

\section*{Acknowledgments}

The author thanks Renaud Raquépas for helpful discussions,
Vojkan Jakšić for drawing his attention to the questions addressed in this work,
and his advisor Ioannis Karatzas for guidance and support. The author is also
grateful to Andrew Barron for insightful comments on an early version of the
manuscript, which motivated the addition of Section~\ref{sec:infinitealphabets}, and to Gaia Pozzoli for
careful comments on an earlier version.

This research was partially supported by the Fonds de Recherche du Québec--Nature
et Technologies (FRQNT).

\setlength{\labelsep}{0.01em}

\end{document}